\documentclass[a4paper,UKenglish]{lipics-v2019}
\usepackage{amsmath,amssymb,amsfonts}
\usepackage{complexity}
\usepackage{thm-restate}
\usepackage{hyperref}

\nolinenumbers

\renewcommand{\A}{\mathcal{A}}
\newcommand{\B}{\mathcal{B}}
\renewcommand{\C}{\mathcal{C}}
\newcommand{\Aff}{\mathsf{Aff}}
\newcommand{\defeq}{:=}
\newcommand{\gen}[1]{\langle #1 \rangle}

\renewcommand{\M}{\mathcal{M}}
\newcommand{\N}{\mathbb{N}}
\renewcommand{\P}{\mathcal{P}}
\newcommand{\PTIME}{\textsf{P}}
\newcommand{\Q}{\mathbb{Q}}
\renewcommand{\R}{\mathbb{R}}
\newcommand{\Reg}{\mathsf{R}}
\newcommand{\rk}{\mathit{rank}}

\renewcommand{\SL}{\mathsf{SL}(2,\mathbb{Z})}
\newcommand{\UT}{\mathsf{UT}}
\newcommand{\UTM}[3]{\begin{pmatrix} #1 & #2 \\ & #3 \end{pmatrix}}

\newcommand{\Ve}[2]{\begin{pmatrix} #1 \\ #2 \end{pmatrix}}
\newcommand{\VVe}[2]{\Ve{#1}{#2}}
\newcommand{\Z}{\mathbb{Z}}
\newcommand{\zero}{\mathbf{0}}

\newcommand{\trans}[1]{\xrightarrow{#1}}
\newcommand{\T}{\mathcal{T}}

\declaretheorem[name=Proposition]{prop}

\title{On Affine Reachability Problems}

\author{Stefan Jaax}{Technische Universit\"at M\"unchen, Germany}{}{}{Supported by an ERC Advanced Grant (787367: PaVeS).}
\author{Stefan Kiefer}{University of Oxford, UK}{}{}{Supported by a Royal Society University Research Fellowship.}

\authorrunning{Stefan Jaax and Stefan Kiefer} 

\Copyright{Stefan Jaax and Stefan Kiefer}



\category{}

\relatedversion{}

\supplement{}

\begin{document}

\maketitle              

\begin{abstract}
We analyze affine reachability problems in dimensions $1$ and~$2$.
We show that the reachability problem for 1-register machines over the integers with affine updates is PSPACE-hard, hence PSPACE-complete, strengthening a result by Finkel et al.\ that required polynomial updates.
Building on recent results on two-dimensional integer matrices, we prove NP-completeness of the mortality problem for 2-dimensional integer matrices with determinants $+1$ and $0$.
Motivated by tight connections with 1-dimensional affine reachability problems without control states, we also study the complexity of a number of reachability problems in finitely generated semigroups of 2-dimensional upper-triangular integer matrices. 
\end{abstract}


\section{Introduction} \label{sec-intro}

\subparagraph*{Counter machines.}

\emph{Counter machines} are abstract models of computation, consisting of a finite control and a set of registers which store numbers.
Upon taking a transition, the machine may perform simple arithmetic on the registers.
There is a great variety of such machines, depending on the domain of the registers ($\Q$, $\Z$, $\N$, \ldots), whether the content of the registers influences the control flow (e.g., via zero tests), and the types of allowed register updates (changes can only be additive, linear, affine, polynomial, \ldots).

As a model of programs with arithmetic, counter machines relate to program analysis and verification.
They also provide natural classes of finitely presented systems with infinitely many states, with a regular-shaped transition structure.
Minsky~\cite{Minsky61} showed that counter machines with nonnegative integer registers, additive updates, and zero tests are Turing-powerful.
\emph{Vector addition systems with states (VASS)}, which are roughly equivalent to \emph{Petri nets}, are a related well-studied model without zero tests.
Reachability in this model is decidable, albeit with very high complexity~\cite{Czerwinski19}.
Recent work~\cite{blondinraskin} considers reachability in certain variants of VASS, including in affine VASS, which are closely related to affine register machines (see below) but have multiple counters.

In this paper, \textbf{we establish the precise complexity of reachability in \emph{affine register machines}.}
These are counter machines with a single integer register (two registers already lead to undecidability~\cite[Chapter~2.5]{ReichertThesis}), no zero tests, and affine register updates; i.e., the transitions are labelled with updates of the form $x := a x + b$, where $x$ stands for the register and $a,b$ are integer coefficients.
Finkel et al.~\cite{FGH} considered a more general model, \emph{polynomial register machines}, with the difference that the updates consist of arbitrary integer polynomials, not just affine polynomials $a x + b$.
The main result of~\cite{FGH} is that reachability in polynomial register machines is \PSPACE-complete.
We show that reachability in affine register machines is also \PSPACE-hard, hence \PSPACE-complete.
Niskanen~\cite{Niskanen17} strengthened the lower bound from~\cite{FGH} in an orthogonal direction, by showing \PSPACE-hardness in the case with polynomial updates but without control states.

As we explain in the following (see also Proposition~\ref{prop-matrix-affine} below), the stateless case is intimately connected with finitely generated monoids over two-dimensional upper-triangular integer matrices.
This leads us to investigate several natural reachability problems in such monoids.

\subparagraph*{Matrix monoids.}
A finite set of matrices $\M \subset \Q^{d \times d}$ generates a monoid $\gen{\M}$ under matrix multiplication, i.e., $\gen{\M}$ is the set of products of matrices from~$\M$, including the identity matrix, which we view as the empty product.
Algorithmic problems about such monoids are hard, often undecidable.
For example, Paterson~\cite{Paterson70} showed in 1970 that the \emph{mortality} problem, i.e., deciding whether the zero matrix is in the generated monoid, is undecidable, even for integer matrices with $d=3$.
It remains undecidable for $d = 3$ with $|\M| = 6$ and for $d = 18$ with $|\M| = 2$; see~\cite{Neary15}. 
Mortality for two-dimensional matrices is known to be \NP-hard~\cite{BellHP12}, but decidability remains a long-standing open problem; see, e.g.,~\cite{PotapovS17}.

Mortality is a special case of the \emph{membership} problem: given $\M$ and a matrix~$T$, is $T \in \gen{\M}$?
Two other natural problems consider certain linear projections of the matrices in~$\gen{\M}$:
The \emph{vector reachability} problem asks, given $\M$ and two vectors $\vec{x}, \vec{y} \in \Q^d$, if there is a matrix $M \in \gen{\M}$ such that $M \vec{x} = \vec{y}$.
Similarly, the \emph{scalar reachability} problem asks if there is a matrix $M \in \gen{\M}$ such that $\vec{y}^T M \vec{x} = \lambda$ holds for given vectors $\vec{x}, \vec{y}$ and a given scalar $\lambda \in \Q$.
For $d=2$ none of these problems are known to be decidable, not even for integer matrices.
In the case $d=2$, mortality \emph{is} known to be decidable when $|\M|=2$ holds~\cite{Bournez2002}, and for integer matrices whose determinants are in $\{-1,0,+1\}$ (see~\cite{NuccioRodaro08}).

Even the case of a single matrix, i.e., $|\M| = 1$, is very difficult; see \cite{OuaknineWorrell12} for a survey.
This case is closely related to the algorithmic analysis of \emph{linear recurrence sequences}, which are sequences $u_0, u_1, \ldots$ of numbers such that there are constants $a_1, \ldots, a_d$ such that $u_{n+d} = a_1 u_{n+d-1} + a_2 u_{n+d-2} + \cdots + a_d u_n$ holds for all $n \in \N$.
In the case $|\M| = 1$ the vector reachability problem is referred to as the \emph{orbit} problem, and the scalar reachability problem as the \emph{Skolem} problem.
The orbit problem is decidable in polynomial time~\cite{KannanLipton}, but the Skolem problem is only known to be decidable for $d \le 4$ (this requires Baker's Theorem)~\cite{OuaknineWorrell12,ChonevOrbit}.

In the following, we do not restrict $|\M|$ but focus on integer matrices in $d=2$.
Recently, there has been steady progress for certain special cases.
It was shown by Potapov and Semukhin~\cite{PotapovS17} that the membership problem for two-dimensional integer matrices is decidable for non-singular matrices.
This result builds on automata-theoretic techniques developed, e.g., in~\cite{ChoffrutKarhumaki05}, where it was shown that the problem of deciding whether $\gen{\M}$ is a group is decidable.
At its heart, this technique exploits the special structure of the group of matrices with determinants $\pm 1$ and its subgroups.
For matrices with determinant~$1$, further results are known, namely decidability of vector reachability~\cite{PotapovSemukhin16} and \NP-completeness of the membership problem~\cite{BellIdentity17}.
If all matrices in~$\M$ have determinant~$1$ and $\M$ is closed under inverses, then $\gen{\M}$ is a group.
In this case, one can decide in polynomial time for a given matrix~$M$ whether $M$ or $-M$ is in~$\gen{\M}$~\cite{GurevichSchupp07}.

Building on these recent results \cite{PotapovSemukhin16,BellIdentity17} \textbf{we prove that the mortality problem for two-dimensional integer matrices with determinants $+1$ or~$0$ is \NP-complete}.
The main result of~\cite{BellHP12} was \NP-hardness for the same problem but allowing also for determinant~$-1$.
Thus, we strengthen the lower bound from~\cite{BellHP12} by disallowing determinant~$-1$, and our subset-sum based proof is considerably simpler.

We then focus on \emph{upper-triangular} integer matrices, i.e., integer matrices of the form $\begin{pmatrix} a & b \\ 0 & c\end{pmatrix}$.
Curiously, decidability of the membership, vector reachability, and scalar reachability problems are still challenging, and indeed open, despite the severe restriction on the matrix shape and dimension.
This class of reachability problems is motivated by its tight connection to (stateless) affine reachability.
For instance, \emph{affine reachability over~$\Q$} reduces to scalar reachability for upper-triangular two-dimensional integer matrices; see Proposition~\ref{prop-matrix-affine}.
Affine reachability over~$\Q$ asks, given a set of affine rational functions in one variable and two rational numbers $x, y \in \Q$, whether $x$ can be transformed into~$y$ using one or more applications of the given functions, chosen nondeterministically.

Whereas affine reachability over~$\Z$ is in~\PSPACE\ by~\cite{FGH}, decidability of affine reachability over~$\Q$ is open.
The problem is related to 
the reachability problem with a single, but only piecewise affine, function (``\emph{piecewise affine maps}''); this problem is not known to be decidable either; see \cite{KPC08,BKP18}.
Variants of piecewise affine reachability, also over~$\Z$, are studied in~\cite{BAmram15}.

\subparagraph*{Organization of the paper.}
In Section~\ref{sec-prelims} we state tight (perhaps folklore) connections between (1) reachability problems in monoids over two-dimensional upper-triangular integer matrices, and (2) reachability problems of one-dimensional affine functions.
We then make the following contributions:
\begin{enumerate}
\item
In Section~\ref{sec-affine-reach-pspace-hard} we show that reachability in affine register machines is \PSPACE-hard, hence \PSPACE-complete.
\item
In Section~\ref{sec-SL} we prove \NP-completeness of the mortality problem for two-dimensional integer matrices with determinants $+1$ or~$0$.
\item
In Section~\ref{sec-2dUTIM} we study the complexity reachability problems in monoids over two-dimensional upper-triangular integer matrices:
\begin{enumerate}
\item
In Section~\ref{sec-determinant-one} we study the case with $\pm1$ on the diagonal.
Establishing a connection with so-called \emph{$\Z$-VASS}~\cite{HaaseHalfon14} allows us to prove \NP-completeness, although we show that the case where all generator matrices have determinant~$-1$ is in~\PTIME\ by a reduction to a linear system of Diophantine equations over the integers.
\item
In Section~\ref{ssec-vec-reachability} we study vector reachability.
We show that the problem is hard for affine reachability over~$\Q$, hence decidability requires a breakthrough,
but the case where the bottom-right entries are non-zero 
is in~\PSPACE.
\item
In Section~\ref{section-membership} we study the membership problem.
If both diagonal entries are non-zero, we show that the problem is \NP-complete, which in turn shows \NP-completeness of the following problem: given $n+1$ non-constant affine functions over~$\Z$ in one variable, can the $n+1$st function be represented as a composition of the other $n$ functions?
The case where only one of the diagonal entries is restricted to be non-zero is decidable in~\PSPACE.
Finally, for the case where both diagonal entries may be~$0$, we establish reductions between membership and scalar reachability, suggesting that decidability of membership would also require a breakthrough.
\end{enumerate}
\end{enumerate}
We conclude in Section~\ref{sec-conclusions}.
For space reasons, some missing proofs are in an appendix.

\section{Preliminaries} \label{sec-prelims}

We write $\Z$ for the set of integer numbers, $\N = \{0,1,2, \ldots\}$ for the set of nonnegative integers, and $\Q$ for the set of rationals.
We write $\UT$ for the set (and the monoid under matrix multiplication) of two-dimensional upper-triangular integer matrices:
\[
 \UT \ \defeq \ \left\{ \begin{pmatrix} x & y \\ 0 & z \end{pmatrix} \;\middle\vert\; x,y,z \in \Z \right\}
\]
We may drop the $0$ in the bottom-left corner and write $\begin{pmatrix} x & y \\ & z\end{pmatrix}$ for matrices in~$\UT$.
Let $\Phi(A)$ be a constraint for $A \in \UT$.
We write 
$\UT[\Phi(A)] \defeq \left\{A \in \UT \mid \Phi(A) \right\}$, e.g.,
$\UT[A_{22} = 1]$ denotes the set of all upper-triangular matrices whose bottom-right
coefficient equals $1$.

For a finite set $\M$ of matrices, we write $\gen{\M}$ for the monoid generated by~$\M$ under matrix multiplication.
In this paper we consider mainly the following reachability problems:
\begin{itemize}
\item \textbf{Membership:} Given a finite set $\M \subseteq \UT$, and a matrix $T \in \UT$, is $T \in \gen{\M}$?
\item \textbf{Vector reachability:} Given a finite set $\M \subseteq \UT$, and vectors $\vec{x}, \vec{y} \in \Z^2$, is there a matrix $M \in \gen{\M}$ such that $M \vec{x} = \vec{y}$?
\item \textbf{Scalar reachability:} Given a finite set $\M \subseteq \UT$, vectors $\vec{x}, \vec{y} \in \Z^2$, and a scalar $\lambda \in \Z$, is there a matrix $M \in \gen{\M}$ such that $\vec{y}^T M \vec{x} = \lambda$?
    We refer to the special case with $\lambda = 0$ as the \emph{$0$-reachability} problem.
\end{itemize}
We write $\Aff(\Z)$ for the set (and the monoid under function composition) of affine functions:
\[
 \Aff(\Z) \ \defeq \ \left\{ x \mapsto a x + b \mid a, b \in \Z \right\} \ \subseteq \ \Z^\Z \text{\quad (where $\Z^\Z = \{f: \Z \to \Z\}$)}
\]
Define $\Aff(\Q)$ similarly, with $\Z$ replaced by~$\Q$.
For a finite set $\A$ of affine functions, we write~$\gen{\A}$ for the monoid (i.e., including the identity function $x \mapsto x$) generated by~$\A$ under function composition.
The motivation to study the matrix reachability problems above is their relationship to the following one-dimensional affine reachability problems:
\begin{itemize}
\item \textbf{Affine membership over $\Z$:} Given a finite set $\A \subseteq \Aff(\Z)$, and a function $f \in \Aff(\Z)$, is $f \in \gen{\A}$?
\item \textbf{Affine reachability over $\Z$:} Given a finite set $\A \subseteq \Aff(\Z)$, and numbers $x,y \in \Z$, is there a function $f \in \gen{\A}$ such that $f(x) = y$?
\item \textbf{Affine reachability over $\Q$:} The same problem with $\Z$ replaced by~$\Q$.
\end{itemize}

These problems are linked by the following proposition.
Recall from the definitions above that the matrices are restricted to be two-dimensional upper-triangular integer matrices.
\begin{restatable}{prop}{matrixaffinerestate}
\label{prop-matrix-affine}

\mbox{}
\begin{enumerate}
\item Affine membership over~$\Z$ is logspace inter-reducible with (matrix) membership restricted to matrices with~$1$ in the bottom-right corner.
\item Affine reachability over~$\Z$ is logspace inter-reducible with vector reachability restricted to matrices with~$1$ in the bottom-right corner and vectors with $1$ in the bottom entry.
\item Affine reachability over~$\Q$ is logspace inter-reducible with $0$-reachability restricted to matrices with non-zero entries in the bottom-right corner and vectors $\vec{x}, \vec{y} \in \Z^2$ such that the bottom entry of~$\vec{x}$ and the top entry of~$\vec{y}$ are non-zero.
\end{enumerate}
\end{restatable}
\begin{proof}
The proof is fairly standard and thus skipped here. It can be found in Appendix~\ref{app:prop-matrix-affine}.
\end{proof}

Simple reductions show that these problems are all \NP-hard:
\begin{prop} \label{prop-NP-hardness}
Membership, vector reachability and $0$-reachability are all \NP-hard, even for matrices with only $1$s on the diagonal, and for $\vec{x} = \Ve{0}{1}$ and for $\vec{y} = \Ve{t}{1}$ (for vector reachability) resp.\ $\vec{y}^T = \begin{pmatrix} 1 & -t \end{pmatrix}$ (for $0$-reachability).
\end{prop}
\begin{proof}
The following problem, \emph{multi-subset-sum}, is known to be \NP-complete; see the comment under ``[MP10] Integer Knapsack'' in~\cite{GJ79}: given a finite set $\{a_1, \ldots, a_k\} \subseteq \mathbb{N}$ and a value $t \in \mathbb{N}$, decide whether there exist coefficients $\alpha_1, \ldots, \alpha_k \in \mathbb{N}$ such that $\sum_{i=1}^k \alpha_i a_i = t$.
Given an instance of multi-subset-sum, construct
\[
 \M \ \defeq \ \left\{ \UTM{1}{a_i}{1} \;\middle\vert\; i \in \{1, \ldots, k\} \right\} \quad \text{and} \quad T \ \defeq \ \UTM{1}{t}{1}\,.
\]
Using the observation that $\UTM{1}{a}{1} \UTM{1}{b}{1} = \UTM{1}{a+b}{1}$ and that, hence, $\gen{\M}$ is commutative, it is straightforward to verify that the instance of multi-subset-sum is positive if and only if $T \in \gen{\M}$.
The proofs for vector reachability and $0$-reachability are similar.
\end{proof}

On various occasions, we make use of the notion of \emph{polynomial register machine}:
Let $\Z[x]$ denote the set of polynomials over~$x$ with integer coefficients.
A \emph{polynomial register machine (PRM)} is a tuple
$\Reg = (Q, \Delta, \lambda)$ where $Q$ is a finite set of \emph{states},
$\Delta \subseteq Q \times Q$ is the \emph{transition relation}, and $\lambda \colon \Delta \rightarrow \Z[x]$ is the \emph{transition labelling function}, labelling each transition with an \emph{update polynomial}. We write $q \trans{p(x)} q'$ whenever $(q, q') \in \Delta$ and $\lambda((q, q')) = p(x)$. The set $\C(\Reg)$ of \emph{configurations} of $\Reg$ is $\C(\Reg) \defeq Q \times \Z$. By $(\trans{}_\Reg) \subseteq \C(\Reg) \times \C(\Reg)$ we denote the \emph{one-step relation} given by
\begin{align*}
(q, a) \trans{}_\Reg (q', b) \quad \Longleftrightarrow \quad q \trans{p(x)} q' \text{ and } b = p(a).
\end{align*}
Let $\mathord{\trans{}^*_\Reg}$ be the reflexive-transitive closure of
$\trans{}_\Reg$.
The next theorem is the main result of~\cite{FGH}:
\begin{theorem}[\cite{FGH}]  \label{thm-fgh}
The following problem is \PSPACE-complete: given a PRM~$\Reg$ and two configurations $(q,a), (q', b) \in \C_\Reg$,
does $(q, a) \trans{}^*_\Reg (q', b)$ hold?
\end{theorem}

Restricting register values to positive numbers up to a given bound and update polynomials to simple increments/decrements leads to the notion of a \emph{bounded counter automaton}.
A bounded one-counter automaton can be specified as a tuple $(Q, b, \Delta)$ where
  \begin{itemize}
    \item $Q$ is a finite set of states,
    \item $b \in \N$ is a global counter bound,
    \item $\Delta$ is the transition relation containing tuples of the form $(q, p, q')$ where
    \begin{itemize}
      \item $q, q' \in L$ are predecessor/successor states,
      \item $p \in [-b, b]$ specifies how the counter should be modified.
    \end{itemize}
  \end{itemize}
  A \emph{configuration} of the automaton consists of a state $q \in Q$ and counter value $c$. We define the set of configurations to be $S = Q \times [0,b]$. For two configurations $(q, c), (q', c')$ we write $(q, c) \trans{} (q', c')$ whenever $(q, p, q') \in \Delta$ for some $p \in \Z$ such that $c' = c + p \in [0, b]$.

By $\trans{*}$ we denote the reflexive-transitive closure of the relation $\trans{}$.

\begin{theorem}[ \cite{fearnley2015reachability}] \label{thm-counter-reach}

The reachability problem for bounded one-counter automata is \PSPACE-complete. This is the following problem: given a bounded one-counter automaton $(Q, b, \Delta)$, a state $q_0$, and a configuration $(q, c) \in S$, does $(q_0, 0) \trans{*} (q, c)$ hold?

\end{theorem}

\section{Reachability in Affine Register Machines}  \label{sec-affine-reach-pspace-hard}
In this section, we show that the reachability problem for PRMs is \PSPACE-complete even when the register updates are restricted to affine functions. We call such PRMs \emph{affine register machines}. We show that reachability in affine register machines is \PSPACE-complete via a reduction from the reachability problem for bounded one-counter automata.

\begin{theorem}
  The following problem is \PSPACE-complete: Given an affine register machine, and configurations $(q, x), (r, y)$, does $(q, x) \trans{*} (r, y)$ hold?
\end{theorem}
\label{thm-pspace-cmp}
\begin{proof}
Membership in \PSPACE\ follows from Theorem~\ref{thm-fgh}. 
It remains to show that the problem is \PSPACE-hard.
Fix a bounded one-counter automaton $\A = (Q, b, \Delta)$. We give a polynomial-time construction of affine register machine $\Reg$ such that $(q_0, 0) \trans{*}_\A (q, c_{\text{tgt}})$ holds for some configurations $(q_0, 0), (q, c_\text{tgt})$ of $\A$ if and only if $(r_0, 0) \trans{*}_\Reg (r, c_\text{tgt})$ holds for some distinct states $r_0$ and $r$ of $\Reg$.

Let $i \in [0, b]$ and $c \in \Z$, and define:
\begin{align*}
K & \defeq 2b + 1 &
K(i, c) & \defeq  (K + 1) c - i \cdot K.
\end{align*}

It can be shown that the following implications hold:
\begin{align}
   i \neq c  \quad\Longrightarrow\quad &  K(i, c)  \not \in [-b, 2b]  \label{impl:one}\\
   i = c \quad\Longrightarrow\quad & K(i, c) = i = c \in [0,b] \label{impl:two}
\end{align}
The (relatively straightforward) derivation of these implications can be found in Appendix~\ref{app:derivation-imp-thm}.

Implications $(\ref{impl:one})$ and $(\ref{impl:two})$ suggest the following (tentative) construction of the PRM $\Reg$ with affine updates: $\Reg$ stores the counter value of the bounded one-counter automaton $\A$ in its register $x$, and it stores the state of $\A$ in its state. When $\Reg$ simulates a transition $(q, c) \trans{} (q', c')$ due to $(q, p, q') \in \Delta$, it does the following:
It guesses $i \in [0, b]$, and performs the updates $x \leftarrow (K+1) \cdot x$, followed by the update $x \leftarrow x -i\cdot K$. If $i = c$ and $x \in [0, b]$, then by $(\ref{impl:two})$ the register value remains unchanged in the interval $[0, b]$; otherwise the updates result in a value outside the interval $[-b, 2b]$ by $(\ref{impl:one})$. Finally, $\Reg$ performs the update $x \leftarrow x + p$ and transitions to state $q'$. Since $p \in [-b, b]$, our sequence of updates maintains the following invariant: once the register value $x$ lies outside the interval $[0, b]$, it remains so forever. Moreover, if the target state $(q, c_\text{tgt})$ is reachable from $(q_0, 0)$ in the counter machine, then we have a corresponding sequence in the affine register machine, and vice versa.

There is one caveat: representing all possible guesses of $i=0,1, \ldots, b$ directly in the transition relation of $\Reg$ would not be polynomial. However, these nondeterministic updates can be represented more succinctly with a slight modification: Let $j = \lceil \log b \rceil + 1$. We first transform the counter machine $\A$ into an equivalent machine $\A'$ with counter bound $B = 2^j - 1 \geq b$ by replacing every transition $(q, p, q') \in \Delta$ with three transitions $(q, p, q'_1)$, $(q'_1, (B-b), q'_2)$, and $(q'_2, -(B-b), q')$ (where $q'_1$ and $q'_2$ are auxiliary intermediate locations).
 Observe that by construction of $\A'$, $(q_0, 0) \trans{*}_\A (q, c)$ holds if and only if $(q_0, 0) \trans{*}_{\A'} (q, c)$. Moreover, the size of $\A'$ is polynomial in the size of $\A$. In order to prove our hardness result, it thus suffices to construct a PRM $\Reg$ of size polynomial in the size of $\A'$, such that  $(q_0, 0) \trans{*}_{\A'} (q, c)$ holds if and only if $(q_0, 0) \trans{*}_\Reg (q, c)$ holds.
 To this end, apply the previous construction of the PRM to $\A'$, but instead of guessing $i \in [0, B]$ directly and computing $x \leftarrow x -i\cdot K$, the PRM uses intermediate auxiliary states $r_0, \ldots r_{j}$, and transitions
 \begin{align*}
  & r_k \trans{x \leftarrow x - 2^k \cdot K} r_{k+1} \qquad \text{ (decrement)} &&
  & r_k \trans{x \leftarrow x} r_{k+1} \qquad \text{ (no update)}
 \end{align*}
 for every $k \in [0, j-1]$. Thus, instead of guessing $i$ and subsequent decrementation of~$x$ by $i \cdot K$, the machine guesses the $j$ binary digits of $i$ in increasing order, and updates the register accordingly after each guess.
 These nondeterministic choices of binary digits represent all updates for values of $i$ in the range $[0, B]=[0, 2^j-1]$.
 The number of states of the resulting PRM is polynomial in the size of the reachability query for $\A'$. This completes the proof.
 \end{proof}

\section{Mortality} \label{sec-SL}

In this section we consider the \emph{mortality} problem: given a finite set $\M \subseteq \Z^{2 \times 2}$ of integer matrices (not necessarily triangular), is the zero matrix~$\zero$ in $\gen{\M}$?
In the upper-triangular case the problem is almost trivial: if there is $M \in \gen{\M}$ with only zeros on the diagonal, there must be $M_1, M_2 \in \M$ with $M_1 = \UTM{0}{b}{c}$ and $M_2 = \UTM{a'}{b'}{0}$---but then $M_1 M_2 = \zero$.
We consider mortality for matrices with determinants $+1, 0$ and prove:
\begin{theorem} \label{thm-mortality}
The mortality problem for two-dimensional integer matrices (not necessarily triangular) with determinants $+1$ or~$0$ is \NP-complete.
It is \NP-hard even if there is one singular matrix and the non-singular matrices are of the form $\UTM{1}{a}{1}$. 
\end{theorem}
Both for the lower and the upper bound we need the following lemma from \cite{Bournez2002}:
\begin{lemma}[{\cite[Lemma~2]{Bournez2002}}] \label{lem-Bournez}
Let $\M \subseteq \R^{2 \times 2}$ be a finite set of matrices.
We have $\zero \in \gen{\M}$ if and only if there are $M_1, \ldots, M_n \in \M$ with $M_1 \cdots M_n = \zero$ and $\rk(M_1) = \rk(M_n) < 2$ and $\rk(M_i) = 2$ for all $i \in \{2, \ldots, n-1\}$.
\end{lemma}
First we prove \NP-hardness:
\begin{proof}[Proof of the \NP-hardness part of Theorem~\ref{thm-mortality}]
Concerning \NP-hardness, the reduction from Proposition~\ref{prop-NP-hardness} for $0$-reachability constructs, given an instance of multi-subset-sum, a set~$\M$ of matrices of the form $\UTM{1}{a}{1}$ (hence of rank~$2$) and a number $t \in \N$ such that the instance of multi-subset-sum is positive if and only if there is $M \in \M$ with $\begin{pmatrix} 1 & -t \end{pmatrix} M \VVe{0}{1} = 0$.
Define the rank-$1$ matrix
$
 T  \defeq  \Ve{0}{1} \begin{pmatrix} 1 & -t \end{pmatrix} = \begin{pmatrix}0 & 0 \\ 1 & -t\end{pmatrix}
$
and set $\M' \defeq \M \cup \{T\}$.
If there is $M \in \M$ with $\begin{pmatrix} 1 & -t \end{pmatrix} M \Ve{0}{1} = 0$, then $\zero = T M T \in \gen{\M'}$.
Conversely, if there is $\zero \in \gen{\M'}$, then, by Lemma~\ref{lem-Bournez}, there is $M \in \M$ with $T M T = \zero$, hence $\begin{pmatrix} 1 & -t \end{pmatrix} M \Ve{0}{1} = 0$.
\end{proof}
We remark that this \NP-hardness proof subsumes the main result of~\cite{BellHP12}, which is slightly weaker in that it allows also for matrices with determinant~$-1$.
For the upper bound we use results from \cite{PotapovSemukhin16,BellIdentity17}.
As usual, define $\SL \defeq \{M \in \Z^{2 \times 2} \mid \det(M) = 1\}$.
\begin{lemma}[{\cite[Lemma~4]{PotapovSemukhin16}}] \label{lem-PotapovSemukhin16}
Let $\vec{x} = \Ve{x_1}{x_2} \in \Z^2$ and $\vec{y} = \Ve{y_1}{y_2} \in \Z^2$ and $M \in \SL$.
If $M \vec{x} = \vec{y}$ then $\gcd(x_1, x_2) = \gcd(y_1, y_2)$.
\end{lemma}
\begin{theorem}[{\cite[Theorem~8, Corollary~9]{PotapovSemukhin16}}] \label{thm-PotapovSemukhin16}
Let $\vec{x}, \vec{y} \in \Z^2$ with $\vec{x} \neq \vec{0}$.
Then one can compute in polynomial time matrices $B, C \in \SL$ such that for every $M \in \SL$ the following equivalence holds: \quad $
 M \vec{x} = \vec{y} \ \Longleftrightarrow \ \text{there is $k \in \Z$ with $M = B \begin{pmatrix} 1 & 1 \\ 0 & 1 \end{pmatrix}^k C$}.$
\end{theorem}
In the following theorem, a \emph{regular expression} describes a set of matrices, so that the atomic expressions describe singleton sets, the operator~$\mathord{\cup}$ is set union, and the operator~$\mathord{\cdot}$ is elementwise multiplication:
\begin{theorem}[{\cite[Corollary 5.1]{BellIdentity17}}] \label{thm-BellIdentity17}
Given a regular expression over matrices in~$\SL$, one can decide in~\NP\ whether the set described by the regular expression intersects with $\{I, -I\}$, where $I$ denotes the identity matrix.
\end{theorem}
Now we can complete the proof of Theorem~\ref{thm-mortality}:
\begin{proof}[Proof of the upper bound in Theorem~\ref{thm-mortality}]
We give an \NP\ procedure.
We guess the matrices $M_1, M_n$ from Lemma~\ref{lem-Bournez}.
Define $\M' \defeq \M \cap \SL$.
We have to verify that there is a matrix $M \in \gen{\M'}$ such that $M_1 M M_n = \zero$.
Let $\vec{x} = \VVe{x_1}{x_2} \in \Z^2$ be a non-zero rational multiple of a non-zero column of~$M_n$ (if $M_n$ does not have a non-zero column, the problem is trivial) such that $\gcd(x_1, x_2) = 1$.
This defines $\vec{x}$ uniquely up to a sign.
Similarly, let $\begin{pmatrix} y_1 & y_2\end{pmatrix} \in \Z^2$ be a non-zero rational multiple of a non-zero row of~$M_1$ such that $\gcd(y_1, y_2) = 1$.
Now it suffices to check whether there is $M \in \gen{\M'}$ such that $\begin{pmatrix} y_1 & y_2 \end{pmatrix} M \vec{x} = 0$.
By Lemma~\ref{lem-PotapovSemukhin16}, this holds if and only if $M \vec{x} \in \{\vec{y}, -\vec{y}\}$ where $\vec{y} := \VVe{-y_2}{y_1}$.
For $\vec{y}$ and $-\vec{y}$, compute the matrices $B_1,C_1$ and $B_2,C_2$ from Theorem~\ref{thm-PotapovSemukhin16}, respectively.
Let $\{A_1, \ldots, A_m\} = \M'$, and note that $A_1^{-1}, \ldots, A_m^{-1} \in \SL$.

We claim that there is $M \in \gen{\M'}$ with $M \vec{x} \in \{\vec{y},-\vec{y}\}$ if and only if there is $i \in \{1,2\}$ with
\begin{equation} \label{eq-reg-exp-intersection}
 B_i\left( \begin{pmatrix} 1 & 1 \\ 0 & 1\end{pmatrix} \cup \begin{pmatrix} 1 & -1 \\ 0 & 1\end{pmatrix} \right)^* C_i \left( A_1^{-1} \cup \cdots \cup A_m^{-1} \right)^* \ \cap\ \{I,-I\} \quad \ne \quad \emptyset\,.
\end{equation}
By Theorem~\ref{thm-BellIdentity17}, this claim completes the proof.

To prove the claim, suppose there is $M \in \gen{\M'}$ with $M \vec{x} = + \vec{y}$ (the negative case is similar).
By Theorem~\ref{thm-PotapovSemukhin16}, there is $k \in \Z$ with $I = B_1 \begin{pmatrix} 1 & 1 \\ 0 & 1 \end{pmatrix}^k C_1 M^{-1}$, implying~\eqref{eq-reg-exp-intersection}.
Conversely, suppose \eqref{eq-reg-exp-intersection} holds for $i=1$ (the case $i=2$ is similar).
Then there are $k \in \Z$ and $M \in \gen{\M'}$ such that $B_1 \begin{pmatrix} 1 & 1 \\ 0 & 1 \end{pmatrix}^k C_1 \in \{M, -M\}$.
By Theorem~\ref{thm-PotapovSemukhin16} it follows that $M \vec{x} = \vec{y}$ or $-M \vec{x} = \vec{y}$.
\end{proof}

\section{Two-Dimensional Upper-Triangular Integer Matrices} \label{sec-2dUTIM}
Motivated by the connections with affine reachability (Proposition~\ref{prop-matrix-affine})
we study in this section the complexity of reachability problems in finitely generated monoids $\gen{\M}$ over two-dimensional upper-triangular integer matrices.
Specifically, we consider membership, vector reachability, and scalar reachability as defined in Section~\ref{sec-prelims}.

\subsection{Determinant \texorpdfstring{$\pm 1$}{+-1}} \label{sec-determinant-one}
In this section we study the case where the monoid $\gen{\M}$ is restricted to matrices with determinants $\pm 1$, i.e., with $\pm 1$ on the diagonal.
In this case, the matrices $M \in \gen{\M}$ are characterized by the sign pattern on the diagonal and the top-right entry.
Our problems become \NP-complete under this restriction, but are in \PTIME\ if the determinants are $-1$.
First we prove the following lemma:
\begin{lemma}\label{lemma_diag_one_np}
Let $\M \subseteq \UT$ be with $\det(M) \in \{1,-1\}$ for all $M \in \M$.
There exists an existential Presburger formula $\varphi(s, a, t)$ that can be constructed in time polynomial in the description of $\M$ such that $\varphi(s, a, t)$ holds if and only if $\begin{pmatrix}s & a \\  & t \end{pmatrix} \in \gen{\M}$.
\end{lemma}
\begin{proof}[Proof sketch]
Note that $\M \subseteq \UT[|A_{11}| = |A_{22}| = 1]$. We reduce the problem whether $\begin{pmatrix}s & a \\  & t \end{pmatrix} \in \gen{\M}$ holds for some $s,t \in \{1,-1\}$ and $a \in \Z$ to a reachability problem on one-dimensional \emph{$\Z$-VASS} (integer vector addition systems with states)~\cite{HaaseHalfon14}. The reachability relation of one-dimensional $\Z$-VASS is known to be effectively definable by an existential Presburger formula in polynomial time; see, e.g.,~\cite{HaaseHalfon14}.
This entails the claim to be shown.
See Appendix~\ref{app-z-vass-induction} for the details.
\end{proof}

\begin{theorem} \label{thm-pm1}
Let $\M \subseteq \UT$ be with $\det(M) \in \{1,-1\}$ for all $M \in \M$.
\begin{enumerate}
\item Membership, vector reachability and scalar reachability are \NP-complete.
\item They are \NP-hard even for $\M \subseteq \UT[A_{11} = A_{22} = 1]$ and for $\M \subseteq \UT[A_{11} = A_{22} = -1]$.
\item They are in \textup{\PTIME} if $\det(M) = -1$ for all $M \in \M$.
\end{enumerate}
\end{theorem}
\begin{proof}[Proof (sketch)]
For item~1 the lower bound follows from Proposition~\ref{prop-NP-hardness}.
The upper bound for membership follows from Lemma~\ref{lemma_diag_one_np} and the folklore result that existential Presburger arithmetic is in $\NP$~\cite{BoroshTreybig76,GathenSieveking78}.
Vector and scalar reachability are easily reduced to membership, as there are only four choices in total for the diagonal entries $s,t$, and this choice together with the input determines the top-right entry uniquely.
This completes the proof of item~1.

Towards item~2, \NP-hardness of the case $\UT[A_{11} = A_{22} = 1]$ follows from the proof of Proposition~\ref{prop-NP-hardness}.
For the case $\UT[A_{11} = A_{22} = -1]$ we adapt this reduction by constructing
\[
 \M \ \defeq \ \left\{ \UTM{-1}{-a_i}{-1} \;\middle\vert\; i \in \{1, \ldots, k\} \right\} \cup \left\{ \UTM{-1}{0}{-1} \right\} \quad \text{and} \quad T \ \defeq \ \UTM{1}{t}{1}\,.
\]
Note that $\UTM{-1}{-a}{-1}  \UTM{-1}{-b}{-1} = \UTM{1}{a+b}{1}$.
The additional (negative identity) matrix ensures that an even number of matrices from~$\M$ can be used to form the product~$T$.
\NP-hardness for vector reachability and $0$-reachability are similar.
This completes the proof of item~2.

Towards item~3, we will give an explicit description of $\gen{\M}$, such that membership can be checked in polynomial time.
In slightly greater detail, we focus on matrix products of even length $M_1 \cdots M_{2 n} \in \gen{\M}$.
These are exactly the matrices in $\gen{\M}$ with determinant~$1$.
The extension to odd-length products (which have determinant~$-1$) will be straightforward, as such products simply arise from even-length products multiplied with a single element of~$\M$.
The even-length products also form a monoid, finitely generated by $\M' \defeq \{M_1 M_2 \mid M_1 \in \M,\ M_2 \in \M\}$, and all matrices in~$\M'$ have $(+1,+1)$ or $(-1,-1)$ on the diagonal.
Clearly, $\M'$ can be computed in polynomial time.
We show in Appendix~\ref{app-thm-pm1} that $\gen{\M'}$ can be characterized by a system of affine Diophantine equations.
It is known that affine Diophantine equations can be solved in polynomial time~\cite[Chapter~5]{Schrijver86}.
The vector reachability and scalar reachability problems (with the restriction on determinants in place) easily reduce to the membership problem, hence are also in~\PTIME.
\end{proof}

\subsection{Vector Reachability} \label{ssec-vec-reachability}
We show:
\begin{theorem}
\label{thm-vec-reachability}
The vector reachability problem for $\UT[A_{22} \neq 0]$ is in $\PSPACE$.
\end{theorem}
\begin{proof}
We construct a nondeterministic Turing machine $\T$ that decides the reachability problem for $\UT[A_{22} \neq 0]$ in polynomial space.
Let $\M \subseteq \UT[A_{22} \neq 0]$ and $\vec{x}, \vec{y} \in \Z^2$ be an instance of the reachability problem, that is, $\T$ has to check whether $M \cdot \vec{x} = \vec{y}$ holds for some $M \in \gen{\M}$.

Assume that $M^{(1)} \cdot \ldots \cdot M^{(k)} \cdot \vec{x} = \vec{y}$ holds for some
$M^{(1)}, \ldots, M^{(k)} \in \M$.
Observe that for all $A \in \M$ and all $z_1, z_2, z_1', z_2' \in \Z$ such that $A  \VVe{z_1}{z_2} = \VVe{z_1'}{z_2'}$, we have
$z'_2 = A_{22} z_2$.
From this observation we conclude:
\begin{enumerate}
	\item If $x_2 \neq 0$, then $y_2 \neq 0$, too, and  the number of indices $1 \leq i \leq k$ s.t.\  $|M^{(i)}_{{22}}| > 1$ is bounded by $\mathcal{O}\left(\text{log}(|y_{2}|) \right)$.
     \item If $x_2 = 0$, then $y_1 = M^{(1)}_{11} \cdot \ldots \cdot M^{(k)}_{11} \cdot x_1$.
	\item If $x_2 = 0$ or $y_2 = 0$, then $x_2 = y_2 = 0$. 
\end{enumerate}

Let us first consider the case where  $x_2 = 0$ or $y_2 = 0$ holds.
In this case, $\T$ rejects the input if $x_2 \neq 0$ or $y_2 \ne 0$.
Otherwise, $\T$ needs to check whether $y_1$ can be written as a product $M^{(1)}_{11} \cdot \ldots \cdot M^{(k)}_{11} \cdot x_1$ for some indices $1, \ldots, k$, which can be done in polynomial space (even in NP), since $k$ can be bounded by $\mathcal{O}(\text{log}(|y_1|))$.

Now consider the case where $|x_2| > 0$ and $|y_2| > 0$ holds.
By the above observation, if the reachability problem has a solution, it can be given by
\begin{align} \label{eq_compose}
\vec{y} = A^{(l+1)} \cdot
          B^{(l)} \cdot
          A^{(l)} \cdot \ldots
          \cdot B^{(2)}
          \cdot A^{(2)}
          \cdot B^{(1)}
          \cdot A^{(1)}
          \cdot \vec{x}\;,
\end{align}
\begin{itemize}
	\item where the length of $l \in \N$ is polynomially bounded in the size of the input,
	\item $B^{(i)} \in \UT[|A_{22}| > 1]\cap \M$ for every  $i$,
	\item $A^{(i)}$ can be written as product of matrices from $\UT[|A_{22}| = 1] \cap \M$.
\end{itemize}

Notice that the matrices from $\UT[|A_{22}| = 1]$ behave like affine update polynomials in a PRM, with the register value stored in the first component of the vector.
This suggests the following approach: $\T$ guesses the sequence of matrices $B=B^{(1)}, \ldots ,B^{(l)}$ and
constructs a PRM $\Reg_B$, whose size is polynomially bounded in the size of the input, such that $(q, x_1) \trans{}^*_{\Reg_B} (q', y_1)$ holds for some fixed states $q, q'$ if the reachability problem has a solution of the form given in~\eqref{eq_compose}. The register machine only needs to store in its states how many of the $B$-matrices have already been applied, plus the current sign of the second vector component reached thus far.
The size is polynomial in the input. The claim then follows by applying Theorem~\ref{thm-fgh}. The details of the construction of $\Reg_B$ can be found in Appendix~\ref{app:details-constr}.
\end{proof}

Without the restriction on $\UT[A_{22} \neq 0]$ the vector reachability problem becomes hard for affine reachability over~$\Q$:
\begin{theorem} \label{thm-hardness-for-Aff-Q}
There is a polynomial-time Turing reduction from affine reachability over~$\Q$ to vector reachability.
\end{theorem}
\begin{proof}
Let an instance of affine reachability over~$\Q$ be given.
We first assume that all input functions are non-constant.
Then we use the reduction from Proposition~\ref{prop-matrix-affine}.3 to obtain an instance of the $0$-reachability problem:
$\M \subseteq \UT[A_{11} \ne 0,\ A_{22} \ne 0]$ and $\Ve{x_1}{x_2}$ and $\Ve{y_1}{y_2}$ with $x_2 \ne 0$ and $y_1 \ne 0$.
(The top-left entries are non-zero as the functions are non-constant.)
Define $T \defeq \UTM{y_1}{y_2}{0}$ and $\M' \defeq \M \cup \{T\}$.
We show that the instance of the $0$-reachability problem is positive if and only if the vector reachability for $\M'$ and $\vec{x}$ and $\vec{0}$ is positive.
Suppose the instance of the $0$-reachability problem is positive.
Then there is $M \in \gen{\M}$ such that $\begin{pmatrix} y_1 & y_2 \end{pmatrix} M \vec{x} = 0$, thus $T M \vec{x} = \vec{0}$, so $\vec{0}$ is reachable from~$\vec{x}$.
Conversely, suppose the instance of the $0$-reachability problem is negative.
Let $M \in \gen{\M}$.
Then $M \vec{x} \ne \vec{0}$, as the bottom component of~$M \vec{x}$ is non-zero.
Since the instance of the $0$-reachability problem is negative, we have $T M \vec{x} = \Ve{t}{0}$ for some $t \ne 0$.
Since the top-left component of all matrices in~$\M' \supseteq \{T\}$ is non-zero, it follows that $M' T M \vec{x} \ne \vec{0}$ holds for all $M' \in \gen{\M'}$.
Thus, $M'' \vec{x} \ne \vec{0}$ holds for all $M'' \in \gen{\M'}$, and so the instance of the vector reachability problem is negative.

Now we allow input functions of affine reachability to be constant.
Suppose the constant functions are $f_i : x \mapsto 0 x + c_i$ for $i \in \{1, \ldots, n\}$ for some $n \in \N$.
It is easy to see that then the affine reachability problem can be solved by removing all $f_i$ from the set of functions and checking affine reachability starting from~$c_i$, where $i \in \{1, \ldots, n\}$, one by one.
These instances can be reduced to vector reachability, as described before.
\end{proof}

\subsection{Membership} \label{section-membership}
In this section we study the membership problem.
As we will see, the difficulty depends on how many $0$s are allowed on the diagonal.
Any product of upper-triangular matrices is non-zero on the top-left (bottom-right, respectively) if and only if all factors are non-zero on the top-left (bottom-right, respectively).
So when we speak of the membership problem for, say, $\UT[A_{11} \ne 0]$, the restriction refers both to $\M$ and the target matrix~$T$.

The case with no $0$s on the diagonal is \NP-complete:
\begin{theorem}\label{thm-membership}
The membership problem for $\UT[A_{11} \neq 0 \land A_{22} \neq 0]$ is \NP-complete.
\end{theorem}
\begin{proof}
The lower bound was shown in Proposition~\ref{prop-NP-hardness}.
For the upper bound, we construct an \NP\ Turing machine.
Fix $\M$ and $T$.
Assume for the moment that $T$ can be written as a product $T = M^{(k)} \cdot \ldots \cdot M^{(1)}$ of matrices from~$\M$.
%
%
Let $l$ be the number of indices $i>1$ where
$M^{(i)} \in \UT[|A_{11}| > 1 \lor |A_{22}| > 1]$ holds.
Since $T_{ii} = \prod_{j=1}^{k} M^{(j)}_{ii}$ holds for both $i \in \{1, 2 \}$, the number
$l$ can be bounded by $\mathcal{O}\left(\text{log}(|T_{11}|) + \text{log}(|T_{22}|)\right)$, and $T$ can be written as
\begin{align}\label{eq-refactor-t}
T = A^{(l+1)} B^{(l)} A^{(l)} \cdot \ldots \cdot B^{(1)} \cdot A^{(1)}M^{(1)}
\end{align} s.t.\
$B^{(j)} \in \UT[|A_{11}| > 1 \lor |A_{22}| > 1] \cap \M$ and
$A^{(j)} \in \gen{\UT[|A_{11}| = |A_{22}| = 1]\cap \M }$ for all $j$.

The Turing machine guesses matrices $B^{(1)}, \ldots, B^{(l)}$ and the matrix $M^{(1)}$ and constructs
in polynomial time an existential Presburger formula $\varphi(t_1, t_2, t_3)$ that satisfies $t_1 = T_{11}$, $t_2 = T_{12}$, $t_3 = T_{22}$ if and only if $T$ can be written as a product of the form given in~\eqref{eq-refactor-t} for the guessed $B^{(i)}$ and $M^{(1)}$. By Lemma~\ref{lemma_diag_one_np}, such a formula $\varphi(t_1, t_2, t_3)$ exists and can be efficiently constructed. The claim then follows from the fact that $\varphi(t_1, t_2, t_3)$ is existential Presburger of size polynomial in the input, and that the existential Presburger fragment is in $\NP$~\cite{BoroshTreybig76,GathenSieveking78}.
\end{proof}
The proof of Proposition~\ref{prop-matrix-affine}.1 with the isomorphism between affine functions over~$\Z$ and upper-triangular matrices with $1$ on the bottom-right shows that non-constant functions correspond to matrices that do not have $0$s on the diagonal.
Hence we have:
\begin{corollary} \label{cor-affine-composition}
Affine membership over~$\Z$ with non-constant functions is \NP-complete.
\end{corollary}

The case with at most one $0$ on the diagonal can be reduced to vector reachability:
\begin{restatable}{theorem}{memberrestate}
\label{thm-membership-diag-one-zero}
The membership problems for $\UT[A_{11} \neq 0]$ and for $\UT[A_{22} \neq 0]$ are in $\PSPACE$.
\end{restatable}
\begin{proof}
We give a proof sketch for $\UT[A_{22} \neq 0]$; the detailed proof can be found in the appendix.
If $T_{11} \neq 0$, then a $\PSPACE$ decision procedure follows from Theorem~\ref{thm-membership}. If $T_{11} = 0$, then the problem reduces to a reachability problem with the additional restriction that some element of $\UT[A_{11} = 0]$ must be included in the matrix product. This problem in turn is decidable in $\PSPACE$ via a straightforward modification of the PRM $\Reg_B$ in the proof of Theorem~\ref{thm-vec-reachability}.
\end{proof}

The general membership problem, without restrictions on the position of $0$s, is related to (variants of) scalar reachability.
Theorems \ref{thm-membership-to-scalar} and~\ref{thm-scalar-to-membership} provide reductions in both ways.
\begin{restatable}{theorem}{scalarrestate}
\label{thm-membership-to-scalar}
Let $s$ be an oracle for the scalar reachability problem. The membership problem is in $\PSPACE^s$.
\end{restatable}
\begin{proof}
Fix some finite $\M \subseteq \UT$ and $T \in \UT$. We give a $\PSPACE^s$ procedure that decides whether $T \in \gen{\M}$ holds.
We make the following case distinction:
\begin{enumerate}
	\item \label{case_1} $T = \mathbf{0}$
	\item \label{case_2} $T \in \UT[A_{11} \neq 0 \lor A_{22} \neq 0]$
	\item \label{case_3} $T \in \UT[A_{11} = A_{22} = 0 \land A_{12} \neq 0]$
\end{enumerate}

In the first case, the membership problem is easy:
if $T = \mathbf{0} \in \gen{\M}$, then there must exist matrices $M_1 \in \UT[A_{11} = 0]\cap\M$ and $M_2 \in \UT[A_{22} = 0] \cap \M$, but then $T = \mathbf{0} = M_1 \cdot M_2$.
The existence of such $M_1, M_2$ is trivial to check.
In the second case, the problem reduces to $T\in \gen{\UT[A_{11} \neq 0]\cap \M}$ or $T\in \gen{\UT[A_{22} \neq 0]\cap \M}$, which is decidable in $\PSPACE$ by Theorem~\ref{thm-membership-diag-one-zero}. In Appendix~\ref{app:scalar} we show that the last case boils down to an instance of scalar reachability by an $\NP$ procedure, and thus can be solved in $\NP^s \subseteq \PSPACE^s$.
\end{proof}

\begin{theorem}\label{thm-scalar-to-membership}
The following sign-invariant version of the scalar-reachability problem is polynomial-time Turing-reducible to the membership problem: given $\M \subseteq \UT$ and column vectors $\vec{x}, \vec{y} \in \Z^2$, does $\vec{y}^T M \vec{x} \in \{-1, 1\}$ hold for some $M \in \gen{\M}$?
\end{theorem}
\begin{proof}

Fix $\M,\vec{x},\vec{y}$.
Let $I$ be the identity, $\A \defeq \M\cap \UT[A_{22} = 0]$, $\B \defeq \M \cap \UT[A_{11} = 0]$, $\C \defeq \left(\M \setminus (\A \cup \B)\right)$,
$Y \defeq \begin{pmatrix}y_1 & y_2 \\ 0 & 0 \end{pmatrix}$, and
$X \defeq \begin{pmatrix} 0 & x_1 \\ 0 & x_2 \end{pmatrix}$. Set $\A' \defeq \A$, if $|y_{1}| = 1$, otherwise set $\A' \defeq \emptyset$; further set $\B' \defeq \B$, if $|x_{2}| = 1$, otherwise set $\B' \defeq \emptyset$.

We obtain the following equivalences:
\begin{align}
	& \exists M \in \gen{\M} \colon \vec{y}^T M \vec{x} \in \{\pm 1\}  & \Leftrightarrow & \label{equiv-zero} \\
	& \exists A \in \A\cup \{I\} , B \in \B\cup \{I\},
	        C \in \gen{\C} \colon \vec{y}^T \cdot (A\cdot C \cdot B) \cdot \vec{x} \in \{\pm 1 \} & \Leftrightarrow \label{equiv-one} & \\
	& \begin{pmatrix}0 & \pm 1 \\ 0 & 0 \end{pmatrix} \in
	\bigcup_{A \in \A' \cup \{Y\}, B \in \B' \cup \{X \}} \gen{\C \cup \{A, B \}}. \label{equiv-two}
\end{align}
 We provide detailed derivations of these equivalences in Appendix~\ref{app-proof-equivalence}. Deciding (\ref{equiv-two}) requires polynomially many queries to a membership oracle where input sizes are polynomial in the description of $\M, \vec{x}, \vec{y}$. This entails the theorem. 
\end{proof}

\section{Conclusion} \label{sec-conclusions}
We have proved \PSPACE-completeness of reachability in affine register machines, and \NP-completeness of the mortality problem over two-dimensional integer matrices with determinants $+1$ or~$0$.

Motivated by their connections to affine reachability, we have studied membership, vector reachability, and scalar reachability for two-dimensional upper-triangular integer matrices.
We have established several complexity results and reductions.
Concerning upper complexity bounds, we have employed a variety of techniques, including existential Presburger arithmetic, $\Z$-VASS, PRMs, and solving linear Diophantine equations over the integers.
We have also established lower bounds, including hardness of vector reachability for affine reachability over~$\Q$, and a connection between membership and scalar reachability.

As open problem, we highlight the precise complexity (between \NP\ and~\PSPACE) of (stateless) affine reachability over~$\Z$. 

\bibliographystyle{plain}
\bibliography{lit}

\clearpage

\appendix
\section{Missing Proof of Proposition~\ref{prop-matrix-affine}}

\label{app:prop-matrix-affine}

\matrixaffinerestate*

\begin{proof}
Items 1 and 2 follow from the isomorphism $\varphi: \Aff(\Z) \to \UT[A_{22} = 1]$ with $\varphi(x \mapsto a x + b) = \UTM{a}{b}{1}$, and the injection $\varphi' : \Z \to \Z^2$ with $\varphi'(x) = \VVe{x}{1}$.
More specifically, we have $\varphi(f_2 \circ f_1) = \varphi(f_2) \varphi(f_1)$ and $\varphi'(f(x)) = \varphi(f) \varphi'(x)$.

For item~3 consider the quotient $\UT/\mathord{\sim}$ of~$\UT$ by the equivalence~$\mathord{\sim}$ with $\UTM{a_1}{b_1}{c_1} \sim \UTM{a_2}{b_2}{c_2}$ if and only if there is $\lambda \in \Q \setminus \{0\}$ with $\UTM{a_1}{b_1}{c_1} = \lambda \UTM{a_2}{b_2}{c_2}$.
We define a similar injection $\varphi: \Aff(Q) \to \UT/\mathord{\sim}$ as above such that $\varphi(x \mapsto \frac{a x + b}{c}) = \UTM{a}{b}{c}$ where $a,b,c \in \Z$ and $c \ne 0$.
It is an isomorphism, as
\[
\varphi\left(x \mapsto \frac{a_2 \frac{a_1 x + b_1}{c_1} + b_2}{c_2}\right) \ = \ \varphi\left(x \mapsto \frac{a_2 a_1 x + a_2 b_1 + b_2 c_1}{c_2 c_1}\right) \ = \ \UTM{a_2}{b_2}{c_2} \UTM{a_1}{b_1}{c_1}\,.
\]
Define also the equivalence $\mathord{\sim'}$ such that $\VVe{p_1}{q_1} \sim' \VVe{p_2}{q_2}$ holds if and only if there is $\lambda \in \Q \setminus\{0\}$ such that $\VVe{p_1}{q_1} = \lambda \VVe{p_2}{q_2}$.
Finally, define the injection $\varphi': \Q \to \Z^2 / \mathord{\sim'}$ with $\varphi\left(\frac{p}{q}\right) = \VVe{p}{q}$, for $p,q \in \Z$ and $q \ne 0$.
Then we have, similarly as in items 1 and~2:
\[
 \varphi'\left(\frac{a \frac{p}{q} + b}{c}\right) \ = \ \UTM{a}{b}{c} \Ve{p}{q} \ = \ \varphi\left(x \mapsto \frac{a x + b}{c}\right) \varphi'\left(\frac{p}{q}\right)
\]
It follows that we have $\frac{a \frac{x_1}{x_2} + b}{c} = \frac{y_1}{y_2}$ if and only if $\Ve{y_1}{y_2} \sim' \UTM{a}{b}{c} \Ve{x_1}{x_2}$, which in turn is equivalent to $\begin{pmatrix} y_2 & -y_1 \end{pmatrix} \UTM{a}{b}{c} \Ve{x_1}{x_2} = 0$.
Item~3 follows.
\end{proof}
\section{Derivation of Implications in Proof of Theorem~\ref{thm-pspace-cmp}}
\label{app:derivation-imp-thm}
We first recall the definitions central to the proof of Theorem~\ref{thm-pspace-cmp}.
Let $i \in [0, b]$ and $c \in \Z$, and define:
\begin{align*}
K & \defeq 2b + 1 &
K(i, c) & \defeq  (K + 1) c - i \cdot K.
\end{align*}

We have to show the following implications:
\begin{align}
   i \neq c  \quad\Longrightarrow\quad &  K(i, c)  \not \in [-b, 2b]  \tag{\ref{impl:one}}\\
   i = c \quad\Longrightarrow\quad & K(i, c) = i = c \in [0,b] \tag{\ref{impl:two}}
\end{align}

Let us first show $(\ref{impl:one})$.
Assume $i \neq c$. We make the following case distinction:
\begin{enumerate}
  \item $c < i$: Then $i = c + \delta$ for some $\delta > 0$. We thus obtain
  \begin{align*}
  K(i, c)
  & = (K+1)c - i \cdot K && \text{ by definition of } K(i, c),\\
  & = (K+1)c - (c+\delta) K  &&\text{ by substituting } i =  c + \delta,\\
  & = c - K \cdot \delta && \\
  & = c - (2b+1) \cdot \delta   && \text{ by definition of } K,\\
  & \leq c - (2b+1) && \text{ since } \delta > 0 \text{ by assumption},\\
  & \leq -b-1 && \text{ since } c < i \leq b.
  \end{align*}
  Hence $K(i, c) \not \in [-b, 2b]$ if $c < i$.
  \item $c > i$:
  \begin{align*}
  K(i, c)
  & = (K+1)c - i \cdot K && \text{ by definition of } K(i, c),\\
  & \geq (K+1)c - (c-1) \cdot K && \text{ since } c > i \geq 0,\\
  & = c + K \\
  & > 2b && \text{ since } c > 0 \text{ and } K = 2b+1.
  \end{align*}
  Hence $K(i, c) \not \in [-b, 2b]$ if $c > i$.
\end{enumerate}
This completes the proof of $(\ref{impl:one})$.
Now, for $(\ref{impl:two})$, observe that setting $i = c$ satisfies $K(i, c) = i = c$.  This completes the proof of $(\ref{impl:two})$.

\section{Details of the Proof of Lemma~\ref{lemma_diag_one_np}}
\label{app-z-vass-induction}
We continue with the proof of Lemma~\ref{lemma_diag_one_np}.

A (one-dimensional) \emph{$\Z$-VASS} can be described as a triple $(Q,\Sigma,\delta)$ where $Q$ is a finite set of states, $\Sigma \subset \Z$ is a finite set of integer numbers, and $\delta \subseteq Q \times \Sigma \times Q$ is a finite transition relation.
A \emph{run from $q_0 \in Q$ to $q_n \in Q$ of length~$n$} is a sequence $q_0 a_1 q_1 a_2 q_2 \cdots a_n q_n$ such that $(q_{i-1}, a_i, q_i) \in \delta$ holds for all $i \in \{1, \ldots, n\}$.
The \emph{value} of such a run is defined as $\sum_{i=1}^n a_i$.
The \emph{$\Z$-VASS reachability} problem is to decide, given a $\Z$-VASS, two states $q, q' \in Q$, and a number $t \in \Z$, whether there is a run from $q$ to~$q'$ with value~$t$.

We give a polynomial reduction from the membership problem to the reachability problem for $\Z$-VASS:
Define $Q \defeq \{(+1,+1),(+1,-1),(-1,+1),(-1,-1)\}$, each state reflecting the diagonal.
A transition $(q, a, q') \in \delta$ corresponds to a multiplication with a matrix $M \in \M$.
More precisely, for each state $(s, t)$ and each $\UTM{s'}{a'}{t'} \in \M$ we add a transition $((s, t), s t t' a', (s s', t t'))$.
We will prove:
\begin{claim} \label{claim-app-Z-VASS}
There exists a run from $(+1,+1)$ to $(s_n,t_n)$ of length~$n$ and value $a \in \Z$ if and only if there are matrices $M_1, \ldots, M_n \in \M$ such that $M_1 \cdots M_n = \UTM{s_n}{t_n a}{t_n}$.
\end{claim}
Claim~\ref{claim-app-Z-VASS} implies that there is $\UTM{s}{a}{t} \in \gen{\M}$ if and only if there is a run from $(+1,+1)$ to $(s,t)$ of value~$t_n a$. Recall that the reachability relation of the $\Z$-VASS is definable in an existential Presburger formula of polynomial size, and thus so is the query $\UTM{s}{a}{t} \in \gen{\M}$.
Lemma~\ref{lemma_diag_one_np} follows.
It remains to prove Claim~\ref{claim-app-Z-VASS}:
\begin{proof}[Proof of Claim~\ref{claim-app-Z-VASS}]
We prove this claim by induction on~$n$.
The case $n=0$ implies $a=0$ and is easy.
For the step, suppose there is a run from $(+1,+1)$ to $(s_n,t_n)$ of length~$n$ and value~$a$.
By the induction hypothesis, there are matrices $M_1 \cdots M_n = \UTM{s_n}{t_n a}{t_n}$.
Consider the run of length $n+1$ and value $a + x$ obtained by extending the previous run by a transition $((s_n,t_n), x, (s_{n+1}, t_{n+1})) \in \delta$.
The definition of~$\delta$ implies that there is matrix $M_{n+1} = \UTM{s'}{a'}{t'} \in \M$ such that $x = s_n t_n t' a'$ and $s_{n+1} = s_n s'$ and $t_{n+1} = t_n t'$.
Hence:
\begin{align*}
M_1 \cdots M_{n+1} &\ = \ \UTM{s_n}{t_n a}{t_n} \UTM{s_n s_{n+1}}{s_n t_n t' x}{t_n t_{n+1}} \ = \ \UTM{s_{n+1}}{t_n t' x + t_{n+1} a}{t_{n+1}} \\
                   &\ = \ \UTM{s_{n+1}}{t_{n+1} (x + a)}{t_{n+1}}
\end{align*}
The other direction (from matrix product to run) is similar.
This proves the claim.
\end{proof}

\section{Details of the Proof of Theorem~\ref{thm-pm1}}
\label{app-thm-pm1}
We continue with the proof of item~3 of Theorem~\ref{thm-pm1}.

Let $\M'_+$ and $\M'_-$ be such that $\M' = \M'_+ \cup \M'_-$, and $\M'_+$ is the set of matrices from~$\M'$ that have $(+1,+1)$ on the diagonal, and $\M'_-$ is defined analogously.
For all $\UTM{-1}{c}{-1} \in \M'_-$, the matrix $\UTM{-1}{-c}{-1}$ is also in~$\M'_-$, as
\begin{equation*} 
\UTM{+1}{a}{-1} \UTM{-1}{b}{+1} = \UTM{-1}{a+b}{-1}  \text{ and }  \UTM{-1}{b}{+1} \UTM{+1}{a}{-1} = \UTM{-1}{-a-b}{-1}\,.
\end{equation*}
The analogous statement for $\M'_+$ holds as well.

Moreover, for any $\UTM{1}{a}{1}, \UTM{1}{b}{1} \in \gen{\M'_+}$ we also have $\UTM{1}{a+b}{1} = \UTM{1}{a}{1} \UTM{1}{b}{1} \in \gen{\M'_+}$.
Let $g \in \N$ denote the gcd of all top-right entries in~$\M'_+$, with $g=0$ in case all those entries are~$0$.
In combination with the observation at the end of the previous paragraph it follows that $\gen{\M'_+} = \UTM{1}{g \Z}{1}$, where here and later in this proof we write $\UTM{a}{B}{c}$ with $B \subseteq \Z$ to denote the set $\left\{\UTM{a}{b}{c} \;\middle\vert\; b \in B \right\}$, and $g \Z$ are the integer multiples of~$g$.
Given two sets of matrices $\A_1, \A_2$, we write $\A_1 \A_2$ to denote the set $\{A_1 A_2 \mid A_1 \in \A_1,\ A_2 \in \A_2\}$.
For $k \in \N$, let $\P_k$ denote the products of elements of~$\M'$ where exactly $k$ factors are from~$\M'_-$.
We have $\P_0 = \UTM{1}{g \Z}{1}$.
Let $S \subset \Z$ such that $\M'_- = \UTM{-1}{S}{-1}$ and define $S^{\oplus k} \defeq \left\{\sum_{i=1}^k s_i \mid s_i \in S\right\}$.
We have:
\begin{align*}
\P_1 &\ = \ \UTM{1}{g \Z}{1} \UTM{-1}{S}{-1} \UTM{1}{g \Z}{1} \ = \ \UTM{-1}{S + g \Z}{-1} \quad \text{and} \\
\P_2 &\ = \ \P_1 \P_1 \ = \ \UTM{-1}{S + g \Z}{-1} \UTM{-1}{S + g \Z}{-1} \ = \ \UTM{1}{S^{\oplus 2} + g \Z}{1} \\
\P_3 &\ = \ \P_2 \P_1 \ = \ \UTM{1}{S^{\oplus 2} + g \Z}{1} \UTM{-1}{S + g \Z}{-1} \ = \ \UTM{-1}{S^{\oplus 3} + g \Z}{-1}
\end{align*}
Continuing this pattern, we see that the set of matrices from~$\gen{\M}$ with $(+1,+1)$ on the diagonal is
\[
 \bigcup_{k \in \N} \P_{2 k} \ = \ \bigcup_{k \in \N} \UTM{1}{S^{\oplus 2 k} + g \Z}{1}\,,
\]
and the set with $(-1,-1)$ on the diagonal is
\[
 \bigcup_{k \in \N} \P_{2 k + 1} \ = \ \bigcup_{k \in \N} \UTM{-1}{S^{\oplus 2 k + 1} + g \Z}{-1}\,.
\]
It remains to show that we can efficiently check membership in such sets.
Let $S = \{s_1, \ldots, s_m\}$.
Recall that we have argued that $S = \{-s_1, \ldots, -s_m\}$.
Suppose $T = \UTM{-1}{t}{-1}$ and we want to check whether $T \in \bigcup_{k \in \N} \P_{2 k + 1}$.
This holds if and only if the linear system given by the two equations
\begin{align*}
 t  =  \sum_{i=1}^m s_i x_i + g y & &  2 k + 1  =  \sum_{i=1}^m x_i
\end{align*}
has an integer solution in the variables $x_i, y, k$.
Indeed, suppose that we have a solution for these two equations with $x_i < 0$ for some~$i$.
Let $s_j = -s_i$.
Then modifying the solution by $x_j := x_j - x_i$ and $k := k - x_i$ and $x_i := 0$ leads to another solution, but $x_i$ is no longer negative.
This can be repeated until all $x_i$ and hence $k$ are nonnegative.

It is known that affine Diophantine equations can be solved in polynomial time~\cite[Chapter~5]{Schrijver86}.
The case where the diagonal entries of $T$ are $(+1,+1)$ is similar.
As remarked in the main text, the extension to odd-length products, which have diagonals $(+1,-1)$ or $(-1,+1)$, is straightforward.
The vector reachability and scalar reachability problems (with the restriction on determinants in place) easily reduce to the membership problem, hence are also in~\PTIME.

\section{Details of the Proof of Theorem~\ref{thm-vec-reachability}}
\label{app:details-constr}
We give a formal specification of the PRM $\Reg_B$ constructed by the Turing machine $\T$ in the proof of Theorem~\ref{thm-vec-reachability}.
Fix some guess $B=B^{(1)}, \ldots , B^{(l)}$ by $\T$.
For a given matrix $M \in \UT$ and a given integer $\alpha$, let $f_{M,\alpha}$ be the affine function given by
$$f_{M,\alpha}(a) \defeq M_{11} \cdot a + M_{12} \cdot \alpha.$$

Assume \eqref{eq_compose} holds for the guessed $B$ and some $A^{(1)}, \ldots, A^{(l+1)}$.
Let $m \in \N$ and $M^{(i)} \in \M$, $1 \leq i \leq m$ be such that
\begin{align*}
 M^{(m)} \cdot M^{(m-1)} \cdot \ldots \cdot M^{(1)}  =
          A^{(l+1)} \cdot
          B^{(l)} \cdot
          A^{(l)} \cdot \ldots
          \cdot B^{(2)}
          \cdot A^{(2)}
          \cdot B^{(1)}
          \cdot A^{(1)}.
 \end{align*}
For $0 \leq i \leq m$, set $\alpha(i) \defeq x_2 \cdot \prod_{j=1}^{i-1}M^{(j)}_{22}$.
Then it is easy to verify that the following holds:
\begin{align*}
  y_1 \ = &\  \left(f_{M^{(m)}, \alpha(m)} \circ f_{M^{(m-1)}, \alpha(m-1)} \ldots \circ f_{M^{(2)}, \alpha(2)}\circ f_{M^{(1)}, \alpha(1)} \right) (x_1) \\
  y_2 \ = &\  M^{(m)}_{22}\cdot \alpha(m) \\
  \alpha(i) \ \in &\ \Big\{\pm x_2 \cdot \prod_{j=1}^k B^{(j)}_{22} \mid k \in \{0, \ldots, l \} \Big\} \text{ for every } 1 \leq i \leq m.
\end{align*}

The PRM $\Reg_B=(Q, \Delta, \lambda)$ can then be defined as
$Q \defeq \left\{\pm x_2 \cdot \prod_{j=1}^k B^{(j)}_{22} \mid k \in \{0, \ldots, l \} \right\}$, $\Delta \defeq Q \times Q$, and the labelling function $\lambda$ is uniquely defined by the following constraints:
\begin{align}
\lambda((\alpha, \alpha')) \subseteq \{f_{M, \alpha} \mid M \in \M \} & \text{ for every } (\alpha, \alpha') \in Q \times Q, \\
f_{M, \alpha} \in \lambda((\alpha, \alpha')) \Longleftrightarrow \alpha' = M_{22} \cdot \alpha & \text{ for every } (\alpha, \alpha') \in Q \times Q,\ M \in \M.
\end{align}
By construction, $(x_2, x_1) \trans{}^*_{\Reg_B} (y_2, y_1)$ holds if \eqref{eq_compose} holds for some matrices $A^{(i)}$, $1 \leq i \leq l+1$. Conversely, if $(x_2, x_1) \trans{f_{M_1, \alpha_1}}_{\Reg_B} \circ \trans{f_{M_2, \alpha_2}}_{\Reg_B} \circ \ldots \circ \trans{f_{M_k, \alpha_k}}_{\Reg_B} (y_2, y_1)$ holds for some $M_1, \ldots, M_k \in \M$, then $\vec{y} = M_k \cdot \ldots \cdot M_1 \cdot \vec{x}$ is a witness for reachability.
Notice that $\Reg_B$ is polynomial in the size of the input, and thus by Theorem~\ref{thm-fgh}, $\T$ can verify in polynomial space whether matrices $A^{(i)}$ exist such that \eqref{eq_compose} holds for the guess $B$, or a permuted subsequence of $B$.

\section{Details of the proof of Theorem~\ref{thm-membership-to-scalar}}
\label{app:scalar}
\scalarrestate*
\label{app:thm-membership-to-scalar}
Recall that we made the following case distinction in the main part:
\begin{enumerate}
  \item \label{case_1} $T = \mathbf{0}$
  \item \label{case_2} $T \in \UT[A_{11} \neq 0 \lor A_{22} \neq 0]$
  \item \label{case_3} $T \in \UT[A_{11} = A_{22} = 0 \land A_{12} \neq 0]$
\end{enumerate}
It remains to prove the last case. We claim that the last case can be solved in $\NP^s$.
Assume for the moment that $T \in \gen{\M}$, that is, $T = M_1 \ldots M_k$ for some matrices $M_i \in \gen{M}$.
Since the product of any two matrices from $\UT[A_{11} = A_{22} = 0 \land A_{12} \neq 0]$ equals the zero matrix, we know that there is at most one index $i$ such that  $M_i \in \UT[A_{11} = A_{22} = 0 \land A_{12} \neq 0]$. Moreover, if no such index $i$ exists, it is easy to verify that there must be some indices $1 \leq i < j \leq k$ such that $M_i \in \UT[A_{22} = 0]$ and $M_j \in \UT[A_{11} = 0]$. Moreover, observe that for every $M \in \UT$, we have $A \cdot M = M_{22}\cdot A$ and $M\cdot A = M_{11}\cdot A$ for every $A \in \UT[A_{11} = 0]$ and every $A \in \UT[A_{22} = 0]$, respectively. From these observations, we conclude that $T \in \gen{\M}$ if and only if one of the following equalities is satisfied for some
$M, M', M'' \in \gen{\M}$:
\begin{align}
  T & = M_{11} \cdot A \cdot M'_{22} && \text{ for some } A \in \M \cap \UT[A_{11} = A_{22} = 0]\,, \label{eq_prod_one} \\
  T & = M_{11} \cdot A \cdot M'' \cdot B \cdot M'_{22} && \text{ for some } A, B \in \M \text{ s.t.\ } A_{22} = 0 \land B_{11} = 0. \label{eq_prod_two}
\end{align}

Note that the absolute values of $M_{11}$ and $M'_{22}$ in (\ref{eq_prod_one}) and (\ref{eq_prod_two}) are bounded by $|T_{12}|$. Moreover, for a product of upper triangular matrices $M = M^{(1)} \cdot \ldots \cdot M^{(k)}$, we have the identity $M_{11} = M^{(1)}_{11} \cdot \ldots \cdot M^{(k)}_{11}$, and thus we may assume without loss of generality that $M$ can be written as a product of matrices from $\M$ whose length is bounded by $\mathcal{O}(\text{log}(|T_{12}|))$, and similarly for $M'$. Now, in order to decide the membership problem, a nondeterministic Turing machine can first guess whether (\ref{eq_prod_one}) or (\ref{eq_prod_two}) holds. If it guesses (\ref{eq_prod_one}), it proceeds to guess the $|T_{12}|$-bounded scalars  $M_{11}$ and $M'_{22}$ and the matrix $A$, verifies the existence of $M$ and $M'$ in nondeterministic polynomial time via guessing the corresponding products of length $\mathcal{O}(\text{log}(|T_{12}|))$, and it finally tests for equality. Thus $(\ref{eq_prod_one})$ can be decided in $\NP$. If the machine guesses $(\ref{eq_prod_two})$, then it proceeds to guess $\alpha=M_{11}$ and $\beta = M'_{22}$, and it also guesses and verifies corresponding matrices $M \in \gen{\M}$ and $M' \in \gen{M'}$ in non-deterministic polynomial time like in the procedure for $(\ref{eq_prod_one})$. The the machine guesses matrices $A \in \M$ and $B \in \M$ such that $A_{22} = 0$ and $B_{11} = 0$ holds. Finally, the machine uses the oracle $s$ to verify the existence of $M'' \in \gen{\M}$ such that $T_{12}=\alpha \begin{pmatrix} A_{11} & A_{12} \end{pmatrix} M''  \begin{pmatrix} B_{12} \\ B_{22} \end{pmatrix} \beta$. Thus, both (\ref{eq_prod_one}) and (\ref{eq_prod_two}) can be decided in $\NP^s$, which establishes membership in $\NP^s \subseteq \PSPACE^s$ for the last case.

\section{Proof of Theorem~\ref{thm-membership-diag-one-zero}}
\memberrestate*
\begin{proof}
We only give the proof for $\UT[A_{22} \neq 0]$; the proof for $\UT[A_{11} \neq 0]$ is symmetric.
Fix some $T \in \UT$ and $\M \subseteq \UT[A_{22} \neq 0]$.
If $T_{11} \neq 0$, then $T \in \gen{\UT[A_{11} \neq 0] \cap \M}$ must hold, which by Theorem~\ref{thm-membership} can be verified in~\NP, and thus in~\PSPACE.
So let us now assume $T_{11} = 0$. Then we have
\begin{align*}
	& & T & =\begin{pmatrix} 0 & T_{12} \\ 0 & T_{22} \end{pmatrix} \in \gen{\M} & \\
	& \Leftrightarrow & T & = M_1 \cdot A \cdot M_2 & \\
    & \Leftrightarrow & \begin{pmatrix}T_{12} \\ T_{22} \end{pmatrix} & =  M_1 \cdot A \cdot M_2 \cdot \begin{pmatrix}0 \\ 1 \end{pmatrix} & \text{ for some } M_1, M_2 \in \gen{\M}, A \in \M \cap \UT[A_{11} = 0]. &
\end{align*}
Thus our problem reduces to the following constrained reachability problem:
Is $\vec{y} = \begin{pmatrix}T_{12} & T_{22} \end{pmatrix}^T$ reachable from $\vec{x} = \begin{pmatrix}0 & 1 \end{pmatrix}^T$ via a product of matrices from $\M$, such that at least one matrix from $\UT[A_{11} = 0]$ occurs in the product? This problem is decidable in $\PSPACE$ via a minor modification of the approach outlined in the proof of Theorem~\ref{thm-vec-reachability}: The polynomial register machine constructed in the proof of Theorem~\ref{thm-vec-reachability} only needs to additionally track whether a matrix from $\UT[A_{11} = 0]$ has been applied.
With $\Reg_B = (Q, \Delta, \lambda)$ from the proof of Theorem~\ref{thm-vec-reachability}, this can be realized with the modified polynomial register machine $\Reg'_B = (Q', \Delta', \lambda')$ where $Q' = \{0, 1\} \times Q$, $\Delta' = Q' \times Q'$ and $\lambda' \colon Q'\times Q' \rightarrow \lambda(Q \times Q)$ is defined minimally such that for every $q, q' \in Q$, every $f_{M, \alpha} \in \lambda((q, q'))$, and every $b \in \{0, 1\}$ the following holds:
\begin{align*}
(b, q) \trans{f_{M, \alpha}} (b, q') & \text{ if } M \in \UT[A_{11} \neq 0], \\
(b, q) \trans{f_{M, \alpha}} (1, q')  & \text{ if } M \in \UT[A_{11} = 0].
\end{align*}
After guessing the matrix sequence $B$ as outlined in the proof of Theorem~\ref{thm-vec-reachability}, a nondeterministic Turing machine only needs to verify $(0, (0, 1)) \trans{*}_{\Reg'_B} (1,(T_{22}, T_{21}))$ in order to obtain a positive witness for the constrained reachability problem, which can be done in $\PSPACE$, and the statement follows.
\end{proof}

\section{Proof of Equivalences in Section~\ref{section-membership}}
\label{app-proof-equivalence}
In this section we prove the equivalences ($(\ref{equiv-zero}) \Leftrightarrow (\ref{equiv-one})$) and $((\ref{equiv-one}) \Leftrightarrow (\ref{equiv-two}))$ from Section~\ref{section-membership}.
Let us recall the definitions of Section~\ref{section-membership}:

Fix $\M \subseteq \UT$, $\vec{x}$, $\vec{y}$.
Define $\A \defeq \M\cap \UT[A_{22} = 0]$, $\B \defeq \M \cap \UT[A_{11} = 0]$,$\C \defeq \left(\M \setminus (\A \cup \B)\right)$,
$Y \defeq \begin{pmatrix}y_1 & y_2 \\ 0 & 0 \end{pmatrix}$, and
$X \defeq \begin{pmatrix} 0 & x_1 \\ 0 & x_2 \end{pmatrix}$.

Further define
\begin{align*}
\A' \defeq \begin{cases} \A & \text{ if } |y_{1}| = 1, \\
                           \emptyset & \text{ otherwise.}
            \end{cases} \\
\B' \defeq \begin{cases} \B & \text{ if } |x_{2}| = 1, \\
                         \emptyset & \text{ otherwise.}
            \end{cases}
\end{align*}

Then the following equivalences hold:
\begin{align}
	& \exists M \in \gen{\M} \colon \vec{y}^T M \vec{x} \in \{\pm 1\}  & \Leftrightarrow & \tag{\ref{equiv-zero}} \\
	& \exists A \in \A\cup \{I\} , B \in \B\cup \{I\},
	        C \in \gen{\C} \colon \vec{y}^T \cdot (A\cdot C \cdot B) \cdot \vec{x} \in \{\pm 1 \} & \Leftrightarrow \tag{\ref{equiv-one}} & \\
	& \begin{pmatrix}0 & \pm 1 \\ 0 & 0 \end{pmatrix} \in
	\bigcup_{A \in \A' \cup \{Y\}, B \in \B' \cup \{X \}} \gen{\C \cup \{A, B \}}. \tag{\ref{equiv-two}}
\end{align}

\noindent{\textbf{Proof of} $\mathbf{(\ref{equiv-zero}) \Leftrightarrow (\ref{equiv-one})}$.}
The direction $(\ref{equiv-one}) \Rightarrow (\ref{equiv-zero})$ is immediate.
Let us now prove $(\ref{equiv-zero}) \Rightarrow (\ref{equiv-one})$.
Assume $(\ref{equiv-zero})$:
\begin{align*}
\exists M \in \gen{\M} \colon \vec{y}^T M \vec{x} \in \{\pm 1\}.
\end{align*}
Fix $M$.
If $M \in \gen{\C}$, then (\ref{equiv-one}) is directly implied.
So let us assume $M \not \in \gen{\C}$ instead. We only need to consider the case where $M$ can be decomposed into
\begin{align}\label{eq-decompose-m}
M = M_1 \cdot A \cdot C \cdot B \cdot M_2
\end{align}
 such that $M_1 \in \gen{\M}, A \in \A, B \in \B, C \in \C$, and $M_2 \in \gen{\M}$. The other cases are similar. For every $M \in \UT$ there exist $\alpha, \beta \in \Z$ such that $M\cdot A = \alpha A$ and $B \cdot M = \beta B$, and thus we can rewrite $(\ref{eq-decompose-m})$ to
 \begin{align}\label{eq-new-decompose-m}
  M = \alpha \cdot A\cdot C \cdot B \cdot \beta
 \end{align}
 for some $\alpha, \beta \in \Z$.
Now, by assumption $\vec{y}^T M \vec{x} \in \{\pm 1\}$ holds, thus we must have $|\alpha| = |\beta| = 1$ in $(\ref{eq-new-decompose-m})$. From this we obtain $\vec{y}^T (A\cdot C \cdot B) \vec{x} \in \{\pm 1\}$. Hence $(\ref{equiv-one})$ holds. This proves the implication.

\medskip

\noindent{\textbf{Proof of} $\mathbf{(\ref{equiv-one}) \Leftrightarrow (\ref{equiv-two})}$.}
Let us first prove the implication $(\ref{equiv-one}) \Rightarrow (\ref{equiv-two})$.
To this end, assume that $(\ref{equiv-one})$ holds:
\begin{align*}
\exists A \in \A\cup \{I\} , B \in \B\cup \{I\},
	        C \in \gen{\C} \colon \vec{y}^T \cdot (A\cdot C \cdot B) \cdot \vec{x} \in \{\pm 1 \}
\end{align*}
Fix $A, B, C$. If $A = B = I$, then (\ref{equiv-two}) is immediate via

\begin{align*}
\vec{y}^T \cdot C \cdot \vec{x} \in \{\pm 1 \} \Leftrightarrow (Y \cdot C \cdot X) = \begin{pmatrix}0 & \pm 1 \\ 0 & 0 \end{pmatrix}.
\end{align*}

So let us consider the most instructive case where $A \in \A \setminus \{I\}$ and $B \in \A \setminus \{I\}$ holds -- the remaining cases can be proved by analogous reasoning.
In this case, we obtain from $(\ref{equiv-one})$ the equality:
\begin{align}\label{eq-above}
	Y \cdot A \cdot C \cdot B \cdot X = y_{1} \cdot A \cdot C \cdot B \cdot x_2 = \begin{pmatrix}0 & \pm 1 \\ 0 & 0 \end{pmatrix}
\end{align}
From $(\ref{eq-above})$ we obtain $|y_1| = |x_2| = 1$, and in particular $\A' = \A$ and $\B' = \B$, and thus:
\begin{align}
A \cdot C \cdot B = \begin{pmatrix}0 & \pm 1 \\ 0 & 0 \end{pmatrix}
\end{align}
where $A \in \A' = \A$ and $B \in \B' = \B$.
This implies $(\ref{equiv-two})$, and we are done showing the implication $(\ref{equiv-one}) \Rightarrow (\ref{equiv-two})$.

Now let us prove the implication $(\ref{equiv-two}) \Rightarrow (\ref{equiv-one})$.
Assume $(\ref{equiv-two})$:

\begin{align*}
\begin{pmatrix}0 & \pm 1 \\ 0 & 0 \end{pmatrix} \in
	\bigcup_{A \in \A' \cup \{Y\}, B \in \B' \cup \{X \}} \gen{\C \cup \{A, B \}}.
\end{align*}

Fix $A\in \A' \cup \{Y\}, B\in \B' \cup \{X \}$ such that
\begin{align}
\label{eq-gen}
\begin{pmatrix}0 & \pm 1 \\ 0 & 0 \end{pmatrix} \in  \gen{\C \cup \{A, B \}}.
\end{align}
By the argument given in the proof of
$(\ref{equiv-zero}) \Rightarrow (\ref{equiv-one})$, $(\ref{eq-gen})$ entails:
\begin{align}
\label{eq-nogen}
\begin{pmatrix}0 & \pm 1 \\ 0 & 0 \end{pmatrix} = A\cdot C \cdot B \text{ for some } C \in \gen{\C}.
\end{align}
Now, if $A = Y$ and $B = X$, $(\ref{equiv-one})$ follows. In the other cases $\A'$ or $\B'$ is non-empty, which by definition of $\A'$ and $\B'$ means that $|y_1| = 1$ or $|x_1| = 1$ must hold, which in turn implies that $Y\cdot A = \pm A$ or  $B\cdot X = \pm B$. Let us consider the most instructive case where $A \neq Y$ and $B \neq X$; the other cases are similar. Then by the previous remarks, we have:
\begin{align*}
Y \cdot A \cdot C \cdot B \cdot X = \pm A \cdot C \cdot \pm B = \begin{pmatrix}0 & \pm 1 \\ 0 & 0 \end{pmatrix}.
\end{align*}
This entails $(\ref{equiv-one})$, and we are done.

\end{document}